\documentclass[letterpaper, 10 pt, conference]{ieeeconf}  

\IEEEoverridecommandlockouts                              
\usepackage{tabularx} 
\usepackage{graphicx} 
\usepackage{cite} 
\usepackage[final]{hyperref} 
\usepackage{amsmath} 
\usepackage{amssymb}  
\usepackage{mathtools, cuted,bm}
\usepackage{caption}
\usepackage{subcaption}
\usepackage{esvect}
\usepackage{lipsum, color}

\newcommand{\K}{\mathbf{K}}
\newcommand{\W}{\mathbf{W}}
\newcommand{\N}{\mathbf{N}}

\newcommand{\V}{\mathbf{V}}

\newcommand{\smallmat}[1]{\left[ \begin{smallmatrix}#1 \end{smallmatrix} \right]}

\newcommand\scalemath[2]{\scalebox{#1}{\mbox{\ensuremath{\displaystyle #2}}}}

\hypersetup{
	colorlinks=true,       
	linkcolor=black,        
	citecolor=black,        
	filecolor=magenta,     
	urlcolor=black         
}

\newtheorem{theorem}{Theorem}
\newtheorem{lemma}[theorem]{Lemma}

\newtheorem{problem}{Problem}
\newtheorem{proposition}{Proposition}

\title{\LARGE \bf
Data-Driven Computation of Robust Invariant Sets and  Gain-Scheduled Controllers for Linear Parameter-Varying Systems
}

\author{Manas Mejari, Ankit Gupta, Dario Piga
\thanks{This work has been accepted for publication in IEEE Control System Letters (LCSS), 2023. }
\thanks{M. Mejari and D. Piga are with IDSIA Dalle Molle
	Institute for Artificial Intelligence, Via la Santa 1, CH-6962 Lugano-Viganello, Switzerland. {\tt\small \{manas.mejari,dario.piga\}@supsi.ch}}
\thanks{A. Gupta is working with vehicle motion planning at Zenseact AB, Gothenburg, Sweden.
        {\tt\small ankit.gupta@zenseact.com}}%
}

\begin{document}
\newcounter{tempEquationCounter}
\newcounter{thisEquationNumber}
\newenvironment{floatEq}
{\setcounter{thisEquationNumber}{\value{equation}}\addtocounter{equation}{1}
\begin{figure*}[!t]
\normalsize\setcounter{tempEquationCounter}{\value{equation}}
\setcounter{equation}{\value{thisEquationNumber}}
}
{\setcounter{equation}{\value{tempEquationCounter}}
\hrulefill\vspace*{4pt}
\end{figure*}

}
\newenvironment{floatEq2}
{\setcounter{thisEquationNumber}{\value{equation}}\addtocounter{equation}{1}
\begin{figure*}[!t]
\normalsize\setcounter{tempEquationCounter}{\value{equation}}
\setcounter{equation}{\value{thisEquationNumber}}
}
{\setcounter{equation}{\value{tempEquationCounter}}
\end{figure*}

}

\maketitle
\thispagestyle{empty}
\pagestyle{empty}

\begin{abstract}
We present a direct data-driven approach to synthesize \emph{robust control invariant} (RCI) sets and their associated gain-scheduled feedback control laws for \emph{linear parameter-varying} (LPV) systems subjected to bounded disturbances. A data-set consisting of a single state-input-scheduling trajectory is gathered from the system, which is directly utilized to compute polytopic RCI set and controllers by solving a  semi-definite program. 
The proposed method does not require an intermediate LPV model identification step. 
Through a numerical example, we show that the proposed approach can generate RCI sets with a relatively small number of data samples when the data satisfies certain excitation conditions. 
\end{abstract}

\section{Introduction}\label{sec:introduction}

A set is referred to as \emph{robust control invariant} (RCI) if, from every initial state within the set, there exists an admissible control input that maintains the state trajectories within the set despite any bounded disturbances acting on the system~\cite{fb}.
For linear parameter-varying (LPV) systems, several algorithms have been proposed to compute RCI set and its associated controllers,  see for \emph{e.g.},
\cite{ag20,Nguyen15, gmfp23auto}. These approaches are \emph{model-based}  in which an LPV model of the  system is assumed to be available. However, obtaining an LPV model of the underlying system poses  several challenges such as structure selection, estimation and validation~\cite{hou13,piga18}.
An inaccurate model can lead to the loss of the invariance property and the violation of constraints during closed-loop operations.

To mitigate the limitations of model-based methods, direct data-driven approaches have emerged as favorable alternatives  
which  circumvent the necessity for model identification.
The direct data-driven algorithm presented in~\cite{bisoffi23},
computes a state-feedback controller from open-loop data to
induce robust invariance in a polyhedral set. However,
this method requires that the set is  known a-priori. Recent contributions offer data-driven techniques that simultaneously compute RCI sets and associated invariance-inducing controllers~\cite{attar23, mejari23, zhong22}. These methods generate RCI sets having zonotopic~\cite{attar23}, polytopic~\cite{mejari23} and ellipsoidal~\cite{zhong22} set representations, however, they are limited to linear time-invariant (LTI) systems. 
In the context of LPV systems, several approaches have focused on direct data-driven algorithms to synthesize LPV controllers, see for \emph{e.g.}, 
LPV input-output controllers for constrained systems~\cite{piga18},  predictive controllers~\cite{verhoek21} and  gain-scheduled controllers~\cite{miller23,verhoek22}. 
In spite of the several methods proposed for control design, to the best of our knowledge, there is no contribution which addresses the problem of RCI set computation for constrained LPV systems in a  data-driven framework.  

In this regard, we develop a direct data-driven approach to compute polytopic RCI sets and gain-scheduled state-feedback controllers for unknown LPV systems subject to state, input, scheduling and disturbance constraints. We built upon the method recently developed by the authors in~\cite{mejari23}. The method  in~\cite{mejari23} is restricted to LTI models, which is extended in this work to LPV systems. A novel data-based condition for invariance is derived, which ensures invariance robustly for all scheduling trajectories within a given set and for all bounded disturbances acting on the system. In particular, a set of \emph{admissible LPV models} is identified, which is compatible with  measured data and disturbance bounds. The invariance is guaranteed robustly for all feasible models in this set. Furthermore, we employ a \emph{gain-scheduled}  controller which uses  available scheduling signal measurements. The set invariance is guaranteed robustly for all scheduling parameters lying in a given  set. Non-linearly between ``for all''  unknown plant models and ``for all'' scheduling parameters is resolved through vertex enumeration of scheduling parameter set and employing full block S-procedure.
Our approach does not require an LPV model of the system but only a single state-input-scheduling trajectory consisting of finite data samples. We show via a numerical example that if the gathered data satisfies certain excitation conditions, then RCI sets and invariance inducing gain-scheduled controllers can be synthesized with relatively  few number of data samples.


\section{Notations and Preliminaries}\label{sec:notation}
We denote by  $\mathbb{I}_m^n \triangleq \{m,\ldots,n\}$, a set of natural numbers between two integers $m$ and $n$, $m\leq n$.
In the paper $*$'s represent matrix entries uniquely identifiable from symmetry. Let $A\in \mathbb{R}^{m \times n}$ be a matrix written according to its $n$ coulumn vectors as  $A = \smallmat{a_1 \ \cdots \ a_n}$, we define vectorization of $A$ as $\Vec{A} \triangleq \smallmat{a^{\top}_{1} \ \cdots \ a^{\top}_{n}}^{\top} \in \mathbb{R}^{mn}$, stacking the columns of $A$. For a finite set $\Theta_v \triangleq \{\theta^{j} \in \mathbb{R}^{n}, j \in \mathbb{I}_{1}^{r}\}$, the convex-hull of $\Theta_v$  is given by,
$\mathsf{conv}(\Theta_v) \triangleq \left\{ \theta\in  \mathbb{R}^{n}: \theta = \sum_{j=1}^{r} \alpha_j \theta^{j}, \mathrm{s.t} \ \sum_{j=1}^{r} \alpha_j =1, \alpha_j \in [0,1]  \right\} $.
$A \otimes B$ denotes the Kronecker product between  $A, B$.

\begin{lemma}[Vectorization] \label{lemma:vectorization}
	For matrices $A \in \mathbb{R}^{k \times l}$, $B \in \mathbb{R}^{l \times m}$, $C \in \mathbb{R}^{m \times n}$ and $D \in \mathbb{R}^{k \times n}$,  the matrix equation $ABC = D$ is equivalent to (see,~\cite[Ex. $10.18$]{abadir05}),  
	\begin{align}
		&(C^{\top} \otimes A) \vv{B} = \vv{ABC} = \vv{D}  \label{eq:vectorize1}
	\end{align}
\end{lemma}

\section{Problem Formulation}\label{sec:prob}

\subsection{Data-generating LPV system and constraints}

The following LPV A-affine discrete-time data-generating system is considered
\begin{equation}\label{eq:system}
	x_{k+1}=\mathcal{A}(p_k)x_k\!+\!B_{o} u_k\!+\!w_k,
\end{equation}
where  $x_k \in \mathbb{R}^{n}$, $u_k \in \mathbb{R}^{m}$, $p_k \in \mathbb{R}^{s} $ and $w_k \in \mathbb{R}^{n}$ denotes the state, control input, scheduling parameter and the (additive) disturbance vectors, at time $k$, respectively.
The   function $\mathcal{A}(p_k)$ 
is assumed to have a linear dependency on the parameter $p_k$ as,
\begin{align}\label{eq:linear_dependence}
	\mathcal{A}(p_k)= \sum_{j=1}^{s} p_{k,j} A_{o}^j, 
\end{align}
where $p_{k,j}$ denotes the $j$-th element of $p_k \in \mathbb{R}^s$ and  $A_{o}^j\in \mathbb{R}^{n \times n}$. The system matrices  $\{A_{o}^j\}_{j=1}^{s}, B_o$ 
are \emph{unknown}. 
Using \eqref{eq:linear_dependence}, the LPV system~\eqref{eq:system} can be written as
\begin{align}\label{eq:true_mat}
	x_{k+1} = \underset{M_o}{\underbrace{\begin{bmatrix}
				A_o^1 \cdots&A_o^s&B_o
	\end{bmatrix}}} \begin{bmatrix}
		p_k \otimes x_k \\
		u_k
	\end{bmatrix} + w_k.
\end{align}
We assume that a data-set  $\{x_k, p_k, u_k\}_{k=1}^{T+1}$ of $T+1$ samples generated from System \eqref{eq:system} is available. The data are arranged in the following matrices,
\begin{subequations}\label{eq:data}
	\begin{align}
		X^{+} & \triangleq [x_2 \quad  x_3 \quad \cdots \quad x_{T+1}] \in \mathbb{R}^{n\times T}, \\
		{X}^{p}_{u} &\triangleq  \begin{bmatrix} p_1 \otimes x_1 \! & \! p_2 \otimes x_2 \!& \! \ldots & p_T \otimes x_T\\  u_1 \! & \! u_2 \!& \! \ldots \!& \! u_T 
		\end{bmatrix} \!\! \in \mathbb{R}^{(ns\!+\!m)\times T},
	\end{align}
\end{subequations}
We consider the following state, input, scheduling and disturbance constraints sets, respectively:
\begin{subequations}
	\begin{align}
		\mathcal{X} \triangleq&\left \{  x: H_xx\leq \mathbf{1}_{n_x} \right \}, \  
		\mathcal{U} \triangleq \left \{  u: H_uu\leq \mathbf{1}_{n_u} \right \}, \label{eq:constr}\\ 
		\mathcal{P} \triangleq &  \mathsf{conv}(\{ p^j \}, j \in \mathbb{I}_{1}^{v_p})\\ 
		\mathcal{W} \triangleq&\left \{  w: -\mathbf{1}_{n_w} \leq H_ww\leq \mathbf{1}_{n_w}\right \},\;\;\; \label{eq:disturbance}
	\end{align}
\end{subequations}
where $H_x, H_u, H_w$,  are given matrices and $\{ p^j \}, j \in \mathbb{I}_{1}^{v_p}$ are given vertices. The  states   are affected by \emph{unknown} disturbance samples $w_k$, however, they are assumed to lie in a bounded known disturbance set $ \mathcal{W}$ for $k \in \mathbb{I}_{1}^{T+1}$.

\subsection{Set of admissible LPV models} and `informative' data
Let us define a set of \emph{admissible model} matrices which satisfy the system equation~\eqref{eq:system} and are compatible with the disturbance 
set $\mathcal{W}$, as follows
\begin{equation}\label{eq:feasible model set}
	\begin{split}
		\mathcal{M}_{T} \triangleq & \left\{M:   
		x_{k+1} - M \begin{bmatrix}
			p_k \otimes x_k\\ 
			u_k
		\end{bmatrix} \in \mathcal{W},  k \in \mathbb{I}_{1}^{T}   \right \}. 
	\end{split}
\end{equation}
where $M= \smallmat{A^1, \ldots A^{s}, \  B} \in \mathbb{R}^{n \times (ns+m)}$. The matrix $M_o$ in \eqref{eq:true_mat} of the true data-generating system belongs to this set.
Using  \eqref{eq:data} and  \eqref{eq:disturbance}, the feasible model set $\mathcal{M}_{T} $ is represented as,
\begin{equation}\label{eq:feasible model set 1}
	\mathcal{M}_{T}\triangleq \left \{M:   
	-\bar{\mathbf{1}} \leq H_w X^{+} - H_w M {X}^{p}_{u} \leq \bar{\mathbf{1}}  \right \},
\end{equation}
with $\bar{\mathbf{1}} \triangleq \smallmat{\mathbf{1}_{n_w} \ \mathbf{1}_{n_w} \cdots \ \mathbf{1}_{n_w}} \in \mathbb{R}^{n_w\times T}$.

We now rewrite the feasible model set $\mathcal{M}_{T}$ in \eqref{eq:feasible model set 1} using the vectorization Lemma~\ref{lemma:vectorization} for $\vv{M} \in \mathbb{R}^{n(ns+m)}$ as follows,
\begin{align}\label{eq:Feasible model set vector}
	\mathcal{M}_{T}
	&\triangleq \left\{\vv{M} : -\mathbf{1}_{Tn_w} \!+\!d  \leq  Z \vv{M} \! \leq \!\mathbf{1}_{Tn_{w}}\!+\!d \right\},
\end{align}
where we define $Z \in \mathbb{R}^{Tn_{w} \times (ns+m)n}$ and $d \in \mathbb{R}^{Tn_w} $ as
\begin{align}\label{eq:Zd}
	Z \triangleq \left( {{X}^{p}_{u}}^{\top} \otimes H_w \right),  \quad  
	d \triangleq \begin{bmatrix}
		H_wx_2 \\ \vdots \\ H_w x_{T+1}
	\end{bmatrix}
\end{align}

\begin{proposition}[Bounded admissible model set]\label{prop:rich_data}
	The admissible model set $\mathcal{M}_{T}$ in \eqref{eq:Feasible model set vector} is a bounded polyhedron if and only if $\mathrm{rank}\left( {{X}^{p}_{u}} \right) = (ns+m)$ and $H_w$ has a full column rank $n$, see~\cite[Fact 1]{bisoffi23}. 
\end{proposition}
From the measured data, the full row rank of ${X}^{p}_{u}$ can be checked. We note that this also relates to the ``informative"  data  and \emph{persistency of excitation} condition for LPV systems~\cite[condition 1]{verhoek22}. If the rank conditions are not met, then for an unbounded set $\mathcal{M}_{T}$, it is challenging to find an RCI set and invariance inducing controller, which will work for all models $M \in \mathcal{M}_T$. 

\subsection{Gain-scheduled invariance inducing control}

We parameterize the controller  to be a state-feedback control law with a linear dependency on the scheduling parameters given as follows
\begin{equation}\label{eq:State Feedback}
	u_k =\mathcal{K}(p_k) x_k,
\end{equation}
where $\mathcal{K}(p_k) =\sum_{j=1}^{s} p_{k,j} \K^j$
with $\K^j \in\mathbb{R}^{m\times n}$ for $j \in \mathbb{I}^{s}_{1}$ are feedback gain matrices. Let $\K = \smallmat{\K^1, \cdots, \K^s } \in \mathbb{R}^{m \times ns}$ such that $u_k = \mathcal{K}(p_k)x_k = \K
\begin{bmatrix}
	p_k \otimes x_k
\end{bmatrix}
$.

From \eqref{eq:feasible model set} and \eqref{eq:State Feedback}, for an admissible model $M \in \mathcal{M}_{T}$,  the  closed-loop dynamics  is given by
\begin{equation}\label{eq:Controlled System}
	x^+ = M \begin{bmatrix}
		I_{ns} \\ 
		\K 
	\end{bmatrix} \begin{bmatrix}
		p \otimes x
	\end{bmatrix} +w,
\end{equation}
where the $k$ dependence is omitted and  $x^+$ denotes the successor state.
We now consider an RCI set $\mathcal{S}$ parameterized as a 0-symmetric polytope\footnote{In the definition of $\mathcal{S}$, we have assumed that $\W$ is invertible, which would be later guaranteed by the LMI conditions for invariance.} defined as follows \vspace{-0.2cm}
\begin{equation}\label{eq:invset1}
	\mathcal{S} \triangleq\left \{ x\in\mathbb{R}^{n} : -\mathbf{1}_{n_c}\leq C\W^{-1}x \leq \mathbf{1}_{n_c} \right \},
\end{equation}
where $C \in \mathbb{R}^{n_c\times n}$ is a fixed  matrix, $\W \in \mathbb{R}^{n\times n}$ is an unknown matrix to be computed. The matrix $C$ is chosen in such a way that $\mathcal{S}$ will be a full dimensional polytope in $\mathbb{R}^n$, the choice of $C$ is discussed in details in \cite[Section 5]{ag20}.

A set~$\mathcal{S}$ is referred to as \emph{robustly invariant} for the system~\eqref{eq:Controlled System}, if for any given $p \in \mathcal{P}$, the following condition is satisfied
\begin{equation}\label{eq:Invariance Condition data based}
	x\in\mathcal{S}  \;\Rightarrow\; x^+\in\mathcal{S},\; \forall M \in \mathcal{M}_{T}, \forall w \in \mathcal{W}.
\end{equation}

By definition the RCI set $\mathcal{S}$ also satisfies the state and input constraints, $\mathcal{S}\subseteq\mathcal{X}$ and $\mathcal{K}\mathcal{S}\subseteq\mathcal{U}$,   given as follows\vspace{-0.2cm}
\begin{align}\label{eq:State Constraint}
	x\in\mathcal{S} \;&\Rightarrow\; x\in\mathcal{X},\\\label{eq:Control Input Constraint} 
	x\in\mathcal{S} \;&\Rightarrow\; u=\mathbf{K}\begin{bmatrix}
		p \otimes x
	\end{bmatrix} \in\mathcal{U} \quad \forall p \in \mathcal{P}
\end{align}

We now formalize the problem considered in this letter:

\begin{problem}\label{prob:Problem Formulation}
	Given the measured state-input-scheduling data $\{x_k, p_k, u_k\}_{k=1}^{T+1}$, the  constraints sets \eqref{eq:constr} and a user-chosen matrix $C$, compute the matrix $\W$ parameterizing the  set  $\mathcal{S}$  in \eqref{eq:invset1} and  feedback controller gains $\K$ such that: 
	$1)$ The invariance condition \eqref{eq:Invariance Condition data based} is satisfied;
	$2)$ The set $\mathcal{S}$  satisfies the state and input constraints \eqref{eq:State Constraint} and \eqref{eq:Control Input Constraint}, respectively. We also aim at maximizing the size of the RCI set $\mathcal{S}$.
\end{problem}\vspace{0cm}

\section{Data-based sufficient LMI conditions for System Constraints and Invariance}\label{sec:LMI}

In order to obtain tractable conditions for the system constraints and the invariance, we define an appropriate coordinate transformation of the state as follows~\cite{ag20}: 
\vspace{-0.1cm}\begin{equation}\label{eq:trans}
	\theta = \W^{-1}x \Leftrightarrow x=\W\theta.\vspace{-0.1cm}
\end{equation}
Based on  this transformation, the set $\mathcal{S}$ in \eqref{eq:invset1} is given as
\vspace{-0.1cm}\begin{equation}
	\label{eq:ninvset}
	\mathcal{S} \triangleq \left \{ \W \theta \in \mathbb{R}^n:\theta \in \Theta \right \},\vspace{-0.1cm}
\end{equation}
where $\Theta$ is a symmetric set defined in the $\theta$-state-space:
\vspace{-0.1cm}\begin{equation}
	\label{eq:Theta}
	\Theta\triangleq\left \{ \theta\in\mathbb{R}^{n}: -\mathbf{1}_{n_c} \leq C\theta \leq \mathbf{1}_{n_c} \right \}.\vspace{-0.1cm}
\end{equation}
Note that, the candidate invariant set $\Theta$ is a \emph{known} 0-symmetric set. 
The choice of $\W$ is the key factor that wholly determines the corresponding polytopic set $\mathcal{S}$ in the original $x$-state-space, and it is our objective to compute $\W$.
The vertices $\left\{\theta^1,\ldots,\theta^{2\sigma }\right\}$ of  $\Theta$ are known for a fixed $C$, and the set  can be  represented as a convex hull of its vertices:
\begin{equation}
	\label{eq:Theta - Convex Hull}
	\Theta=\mathsf{conv}\left(\left\{\theta^1,\ldots,\theta^{2\sigma }\right\}\right),
\end{equation}
where $\sigma >0 $ is a known  integer determined by   $C$.
\subsection{System constraints}
Based on the transformation defined in \eqref{eq:trans}, the state and input inequality constraints \eqref{eq:constr} can be expressed in the $\theta$-state-space. 
When the  constraints are met at all the vertices $\{\theta^i\}_{i=1}^{2\sigma}$, of the convex set $\Theta$, it guarantees their fulfillment for the entire set $\Theta$ due to the convexity property.
The state constraints \eqref{eq:State Constraint} in terms of $\W$ at the vertices of $\Theta$  are given as follows:
\vspace{-0.1cm}\begin{equation}\label{eq:Tractable State Constraints}
	H_x\W\theta\leq \mathbf{1}_{n_x},\forall\theta\in\Theta \;\Leftrightarrow\; H_x\W\theta^i\leq \mathbf{1}_{n_x},\;i\!\in \mathbb{I}_{1}^{2\sigma} \vspace{-0.1cm}
\end{equation}
In order to express the control input constraints in terms of $\W$, we introduce the matrix variables $\N^{l} \in \mathbb{R}^{m \times n} $ for $l \in \mathbb{I}_{1}^{s}$ as follows:
\begin{equation}
	\label{eq:N}
	\N \triangleq \begin{bmatrix}
		\N^1 & \cdots & \N^s
	\end{bmatrix} = 
	\begin{bmatrix}
		\K^1 \W & \cdots & \K^s \W
	\end{bmatrix}
	\in \mathbb{R}^{m \times ns} 
\end{equation}
with $\K^l = \N^l \W^{-1} $.
The control input constraints in \eqref{eq:Control Input Constraint} are then given as 
\begin{align}\label{eq:Control Input Constraints}
	&H_u \K \begin{bmatrix}
		p_k \otimes \W \theta
	\end{bmatrix}   \leq \mathbf{1}_{n_u},
\end{align}
where, 
\begin{align}\label{eq:K_N_identity}
	\K \begin{bmatrix}
		p_k \otimes \W \theta
	\end{bmatrix}  &=
	\begin{bmatrix}
		\K^1\W p_{k,1}+\cdots+\K^s \W p_{k,s} 
	\end{bmatrix}\theta \nonumber \\  &=  \underset{\N}{\underbrace{\begin{bmatrix}
				\K^1\W & \cdots & \K^s \W  
	\end{bmatrix}}} \begin{bmatrix}
		p_k \otimes  \theta
	\end{bmatrix}
\end{align}
The input constraints \eqref{eq:Control Input Constraints} can be written as
\begin{align}
	&H_u \N \begin{bmatrix}
		p \otimes  \theta
	\end{bmatrix}\leq \mathbf{1}_{n_u},\forall (\theta,p) \in (\Theta,\mathcal{P}) \;\Leftrightarrow\; \label{eq:Tractable Control Input Constraints0} \\ &H_u \N \begin{bmatrix}
		p^j \otimes  \theta^i
	\end{bmatrix}\leq \mathbf{1}_{n_u},\; i \in \mathbb{I}_{1}^{2\sigma}, j \in \mathbb{I}_{1}^{v_p} \label{eq:Tractable Control Input Constraints}
\end{align}
Note that   \eqref{eq:Tractable Control Input Constraints0} and \eqref{eq:Tractable Control Input Constraints} are equivalent as  the Kronecker product map is linear in each of its  arguments $(p, \theta)$  and the sets $\mathcal{P}, \Theta$ and $\mathcal{U}$ are convex.

\subsection{Invariance condition}
We now write the system dynamics  in the $\theta$-state-space for an admissible LPV model $M \in \mathcal{M}_{T}$ and for all $w \in \mathcal{W}$, $p \in \mathcal{P}$. By using \eqref{eq:trans} and \eqref{eq:K_N_identity},  the closed-loop dynamics \eqref{eq:Controlled System} can be written as 
\begin{align}\label{eq:closeloop}
	\W \theta^{+}&=M \begin{bmatrix}
		p \otimes \W \theta \\ 
		\K \begin{bmatrix}
			p \otimes  \W \theta
		\end{bmatrix}
	\end{bmatrix}  +w =M \begin{bmatrix}
		I_s p \otimes \W \theta \\ 
		\N \begin{bmatrix}
			p \otimes  \theta
		\end{bmatrix}
	\end{bmatrix}  +w, \nonumber \\
	&=M \begin{bmatrix}
		\bar{\W} \\ 
		\N 
	\end{bmatrix} \begin{bmatrix}
		p \otimes  \theta
	\end{bmatrix} +w, 
\end{align}
where $\bar{\W} \triangleq I_s \otimes \W$ and the last equality follows from the mixed-product property of the Kronekar product. 
Based on the closed-loop dynamics \eqref{eq:closeloop}, we will now state and prove two equivalent invariance conditions in the $\theta$ state-space.  

\begin{lemma}\label{lem:equivalence}
	If the set $\Theta$ in \eqref{eq:Theta} is robustly invariant for system \eqref{eq:closeloop} then the following two statements are equivalent:
	\begin{enumerate}
		\item[$(i)$]  for all $\theta \in \Theta$, for any given  $p \in \mathcal{P}$, $\forall (w, M) \in (\mathcal{W}, \mathcal{M}_{T})$,
		\begin{equation}\label{eq:inv_cond_all}
			\theta^{+} = \left( \W^{-1}  M\begin{bmatrix}
				\bar{\W} \\ 
				\N 
			\end{bmatrix}\begin{bmatrix}
				p \otimes  \theta
			\end{bmatrix}  + \W^{-1}w \right) \in \Theta
		\end{equation}
		\item[$(ii)$] for each vertex $\theta^{i}, \ i \in \mathbb{I}_{1}^{2\sigma}$ of the set $\Theta$,  and for each vertex $p^{j}, \ j  \in \mathbb{I}_{1}^{v_p}$ of the set $\mathcal{P}$,  $\forall(w,  M) \in (\mathcal{W}, \mathcal{M}_{T})$,  
		\begin{equation}\label{eq:inv_cond_vertices}
			\!\!{\theta^{i,j}}^{+} \!\!=\!\! \left(\! \W^{-1}  M\begin{bmatrix}
				\bar{\W} \\ 
				\N 
			\end{bmatrix} \begin{bmatrix}
				p^j \otimes  \theta^i
			\end{bmatrix}  \!+\! \W^{-1}w \!\right) \in \Theta
		\end{equation}
	\end{enumerate}
\end{lemma}
\vspace{0.1cm}
See Appendix~\ref{sec:appendix} for the  proof of Lemma~\ref{lem:equivalence}. In the rest of the paper, we will consider  condition \eqref{eq:inv_cond_vertices} for ensuring robust invariance of the set $\Theta$.

\subsection{LMI Condition for Data-Driven Invariance}
We will now present a data-based sufficient LMI condition  to compute $\W$ and  an associated LPV state-feedback controller such that 
the set $\Theta$ is rendered invariant.  
The  closed-loop dynamics \eqref{eq:inv_cond_vertices} at the vertices $\theta^{i}, p^j$ can be written as,
\begin{align}\label{eq:G_dynamics}
	\W {\theta^{i,j}}^{+} =   \underbrace{\left( \left( \begin{bmatrix}
			\bar{\W} \\ 
			\N 
		\end{bmatrix} \begin{bmatrix}
			p^j \otimes  \theta^i
		\end{bmatrix}\right)^{\top} \otimes I_n \right)}_{\mathcal{G}\left(\W, \N, p^j,  \theta^{i}\right)}  \vv{M} + w.    
\end{align}
where we have used  the vectorization identity in \eqref{eq:vectorize1}.
In order to obtain less conservative LMI conditions, we  introduce new matrix variables $\V_{ijk} \in \mathbb{R}^{n \times n}$ and signals $\xi_{ijk} = \V_{ijk}^{-1} \W {\theta^{i,j}}^{+}$, for $k \in \mathbb{I}_{1}^{n_c}$, $i \in \mathbb{I}_{1}^{2\sigma}$, $j \in \mathbb{I}_{1}^{v_p}$, and express the dynamics~\eqref{eq:G_dynamics} as follows,
\begin{equation}\label{eq:new_variables}
	\mathcal{G}\left(\W, \N, p^j, \theta^{i} \right)\vv{M} +w - \V_{ijk} \xi_{ijk} = 0.
\end{equation}
From the set definition in \eqref{eq:Theta} the invariance condition   \eqref{eq:inv_cond_vertices} is given as, 
for all $k \in \mathbb{I}_{1}^{n_c}, \; i \in \mathbb{I}_{1}^{2\sigma} \; j \in \mathbb{I}_{1}^{v_p} $, 
\begin{equation}\label{LMI Starting Point}
	1 - (e_k^{\top}C{\theta^{i,j}}^{+})^{2}\geq 0,\; \forall w \in \mathcal{W}, \; \forall \vv{M} \in \mathcal{M}_{T},
\end{equation} 
where  $e_k$ is the $k$-th column vector of the identity matrix $I_{n_c}$.
By substituting~\eqref{eq:new_variables} in \eqref{LMI Starting Point} we obtain,
\begin{equation}\label{LMI Starting Point new}
	1 \!-\! (e_k^{\top}C\W^{-1}\V_{ijk}\xi_{ijk})^{2}\geq 0,\; \forall w \in \mathcal{W}, \; \forall \vv{M} \in \mathcal{M}_{T},
\end{equation}
Following the S-procedure~\cite{cs97},  \eqref{LMI Starting Point new} is multiplied 
by a positive scalar variable $\bm{\phi}_{ijk} >0$ and  the left hand side is lower bounded by a term that is guaranteed to be non-negative for all $w \in \mathcal{W}, \; \vv{M} \in \mathcal{M}_{T}$  as follows,
\begin{multline} \label{eq:relax_ineq_dilated}
	\bm{\phi}_{ijk}(1 - (e_k^{\top}C\W^{-1} \V_{ijk}\xi_{ijk})^2) \geq \\
	2\xi_{ijk}^{\top}\!\underbrace{\left(\mathcal{G}(\W, \N, p^j,\theta^{i})\vv{M}+w- \V_{ijk} \xi_{ijk} \right)}_{0} \\ + 
	\underbrace{\left( (\mathbf{1}+d)-Z\vv{M} \right)^{\top}\bm{\Lambda}_{ijk}\left((\mathbf{1}-d)+Z\vv{M}\right)}_{\geq0} \\
	+\underbrace{(\mathbf{1}+H_w w)^{\top}\bm{\Gamma}_{ijk}(\mathbf{1}-H_w w)}_{\geq0},
\end{multline}
with 
$\bm{\Lambda}_{ijk} \in \mathbb{D}^{Tn_w}_{+},  \bm{\Gamma}_{ijk} \in \mathbb{D}^{n_w}_{+}$, being diagonal matrix variables having non-negative entries.
Note that the right hand side of \eqref{eq:relax_ineq_dilated} is non-negative, which can be verified based on \eqref{eq:G_dynamics} and the set definitions  $\mathcal{W}$, $\mathcal{M}_{T}$ in \eqref{eq:disturbance}, \eqref{eq:Feasible model set vector} respectively.
We rewrite \eqref{eq:relax_ineq_dilated} into the following quadratic form:
\begin{equation}\label{eq:quad_form}
	\varkappa^{\top}\mathcal{P}_{ijk}(\W,\N,\bm{\Lambda}_{ijk},\bm{\Gamma}_{ijk},\bm{\phi}_{ijk}, \V_{ijk})\varkappa\succcurlyeq 0,\;\forall \varkappa, 
\end{equation}
where $\varkappa^{\top}=\begin{bmatrix}
	1 &\vv{M}^{\top} & w^{\top} &-\xi_{ijk}^{\top}
\end{bmatrix}$ and $\mathcal{P}_{ijk}$ is a symmetric matrix. Thus, a sufficient  invariance condition is the LMI $\mathcal{P}_{ijk}\succcurlyeq0$, \emph{i.e.},
\begin{equation}\label{eq:LMI invariance condition dilated}
	\scalemath{0.85}{
		\begin{bmatrix}
			\bm{r}_{ijk} & -d^{\top}\bm{\Lambda}_{ijk}Z & \bm{0}  & \bm{0}\\ 
			* &  Z^{\top}\bm{\Lambda}_{ijk}Z   &\bm{0} & \mathcal{G}^{\top}\left(\W, \N, p^j, \theta^{i}\right)\\ 
			* & * & H^{\top}_w \bm{\Gamma}_{ijk} H_w \!& I_n\\
			*&*&*&\V_{ijk}\!+\!\V_{ijk}^{\top}\!-\!\V_{ijk}^{\top}\mathcal{L}_{ijk} \V_{ijk}
		\end{bmatrix} \succcurlyeq 0,}
\end{equation} 
where $\mathcal{L}_{ijk} \triangleq \bm{\phi}_{ijk} \W^{-\top}C^{\top}e_{k}e_{k}^{\top}C \W^{-1}$ and  $\bm{r}_{ijk} \in  \mathbb{R}$, $\mathcal{G}(\W, \N, p^j, \theta^{i}) \in \mathbb{R}^{n \times (ns+m)n}$ are  defined as follows:
\begin{subequations}
	\begin{align}
		&\bm{r}_{ijk} \triangleq \bm{\phi}_{ijk}\!-\!\mathbf{1}^{\top}\!\bm{\Lambda}_{ijk}\mathbf{1}-\!\mathbf{1}_{n_w}^{\top}\!\bm{\Gamma}_{ijk}\mathbf{1}_{n_w} +d^{\top}\bm{\Lambda}_{ijk}d, \label{eq:rij} \\
		&\mathcal{G}\left(\W, \N,  p^j, \theta^{i} \right) \triangleq  \left( \begin{bmatrix}
			\bar{\W} \\ 
			\N 
		\end{bmatrix} \begin{bmatrix}
			p^j \otimes  \theta^i
		\end{bmatrix}\right)^{\top} \otimes I_n, \label{eq:G}
	\end{align}   
\end{subequations}
Note that the block $(4,4)$ in \eqref{eq:LMI invariance condition dilated} exhibits  nonlinear dependence on the variables $ \bm{\phi}_{ijk}, \V_{ijk} $ and $\W$. To resolve this non-linearity we will introduce new matrix variables. 

\begin{theorem}[Data-based LMI conditions for invariance]\label{thm:data-based invariance dilated}
	Given available data $(X^+, {X}^{p}_{u})$ and a user-specified fixed matrix $C \in \mathbb{R}^{n_c \times n}$, if there exists variables $\W \in \mathbb{R}^{n \times n}$, $\N \in \mathbb{R}^{m \times ns}$, and   
	$\{ \bm{\phi}_{ijk} \in \mathbb{R}_{+},  \bm{\Lambda}_{ijk} \in \mathbb{D}^{Tn_w}_{+}, \bm{\Gamma}_{ijk} \in \mathbb{D}^{n_w}_{+}, \bm{X}_{ijk}, \V_{ijk} \in \mathbb{R}^{n \times n}\}$ that satisfy the following LMIs for $k \in \mathbb{I}_{1}^{n_c}$, $i \in \mathbb{I}_{1}^{2\sigma}$ and $j \in \mathbb{I}_{1}^{v_p}$,
	\begin{equation}\label{eq:Dilated LMI 1}
		\begin{bmatrix}
			\W^{\top} + \W - \bm{X}_{ijk} & \bm{\phi}_{ijk} C^{\top}e_{k} \\
			\bm{\phi}_{ijk} e_{k}^{\!\top}C   & \bm{\phi}_{ijk}
		\end{bmatrix} {\succcurlyeq} 0.
	\end{equation}
	\begin{equation}\label{eq:Dilated LMI 2}
		\scalemath{0.85}{
			\begin{bmatrix}
				\bm{r}_{ijk} & -d^{\top}\bm{\Lambda}_{ijk}Z & \bm{0}  & \bm{0} & \bm{0}\\ 
				* &  Z^{\top}\bm{\Lambda}_{ijk}Z   &\bm{0} & \mathcal{G}^{\top}\left(\W, \N, p^j, \theta^{i}\right) & \bm{0} \\ 
				* & * & H^{\top}_w \bm{\Gamma}_{ijk} H_w \!& I_n & \bm{0}\\
				*&*&*&\V_{ijk}\!+\!\V_{ijk}^{\top} & \V^{\top}_{ijk} \\
				*&*&*& * & \bm{X}_{ijk}
			\end{bmatrix} \succcurlyeq 0,}
	\end{equation}
	where, $\bm{r}_{ijk}$, $\mathcal{G}(\W, \N, p^j, \theta^{i})$ are as defined in \eqref{eq:rij}, \eqref{eq:G},
	then, the state feedback controller gain is obtained as $\K^l = \N^l \W^{-1}$ for $l \in \mathbb{I}_{1}^{s}$ which renders the set $\mathcal{S}$ in \eqref{eq:ninvset} robust invariant.  \end{theorem}
\begin{proof}
	Let us introduce a new matrix variable $\bm{X}_{ijk}= \bm{X}^{\top}_{ijk} \succ 0$ such that
	\begin{equation}\label{eq:Xi Original}
		\bm{X}_{ijk}^{-1}\!-\!\mathcal{L}_{ijk} {\succ} 0 \Leftrightarrow  \bm{X}_{ijk}^{-1}\!-\!\bm{\phi}_{ijk} \W^{-\top}C^{\top}e_{k}e_{k}^{\!\top}C \W^{-1}{\succ} 0, 
	\end{equation} 
	in order to resolve the non-linearity in the block $(4,4)$ of \eqref{eq:LMI invariance condition dilated}.
	By applying Schur complement to \eqref{eq:Xi Original} we obtain,
	\begin{equation}\label{eq:Xi Schur}
		\begin{bmatrix}
			\bm{X}^{-1}_{ijk} & \bm{\phi}_{ijk} \W^{-\top}C^{\top}e_{k} \\
			\bm{\phi}_{ijk} e_{k}^{\!\top}C \W^{-1}  & \bm{\phi}_{ijk}
		\end{bmatrix} {\succ} 0.
	\end{equation}
	Using congruence transformation matrix $\mathrm{diag}\{\W, I_n\}$, LMI in \eqref{eq:Xi Schur} can be rewritten as
	\begin{equation}\label{eq:Xi lmi}
		\begin{bmatrix}
			\W^{\top}\bm{X}^{-1}_{ijk}\W & \bm{\phi}_{ijk} C^{\top}e_{k} \\
			\bm{\phi}_{ijk} e_{k}^{\!\top}C   & \bm{\phi}_{ijk}
		\end{bmatrix} {\succ} 0.
	\end{equation}
	In order to resolve the nonlinear dependence in the $(1,1)$ block of the left hand side matrix in \eqref{eq:Xi lmi}, we use,
	\begin{align}\label{eq:linearization W}
		&\W^{\!\top}\bm{X}^{-1}_{ijk}\W \!\! 
		=\!\! (\W \! \!-\! \!  \bm{X}_{ijk})^{\! \top} \bm{X}^{-1}_{ijk} \! (\W \! -\!   \bm{X}_{ijk}) \! + \! \W \!\!+\!\! \W^{\!\top} \! \!-\! \!  \bm{X}_{ijk} \nonumber \\
		&\succcurlyeq \W + \W^{\top} -  \bm{X}_{ijk}
	\end{align}
	By  substituting  $\W^{\top}\bm{X}^{-1}_{ijk}\W$ in \eqref{eq:Xi lmi} with $\W + \W^{\top} -  \bm{X}_{ijk}$, we obtain a sufficient LMI condition for \eqref{eq:Xi lmi} as given in \eqref{eq:Dilated LMI 1}. Thus, proving the first LMI condition~\eqref{eq:Dilated LMI 1} stated in Theorem~\ref{thm:data-based invariance dilated}.
	
	From \eqref{eq:Xi Original}, the condition \eqref{eq:LMI invariance condition dilated} can be rewritten as
	\begin{equation}\label{eq:invariance condition dilated Xi}
		\scalemath{0.85}{
			\begin{bmatrix}
				\bm{r}_{ijk} & -\!d^{\top}\bm{\Lambda}_{ijk}Z & \bm{0}  & \bm{0}\\ 
				* &  Z^{\top}\bm{\Lambda}_{ijk}Z   &\bm{0} & \mathcal{G}^{\!\top}\left(\!\W, \N, p^j, \theta^{i}\!\right)\\ 
				* & * & H^{\top}_w \bm{\Gamma}_{ijk} H_w \!& I_n\\
				*&*&*&\V_{ijk}\!+\!\V_{ijk}^{\!\top}\!\!-\!\!\V_{ijk}^{\!\top}\bm{X}^{-1}_{ijk} \V_{ijk}
			\end{bmatrix} \!\!\succcurlyeq \!\!0,}
	\end{equation} 
	which, after applying the Schur complement, we obtain  the second LMI condition \eqref{eq:Dilated LMI 2}. 
\end{proof}

\section{SDP program for volume maximization of the RCI set}\label{sec:SDP}
The volume of the polytopic set $\mathcal{S}$ in \eqref{eq:invset1} is proportional to the determinant $|\mathrm{det}(\W)|$~\cite{ag20}, for a given $C$. 
We  present an \emph{iterative} determinant
maximization algorithm to maximize the size of $\mathcal{S}$. At the $q$-th iteration, let 
$W^q$ and ${X}^q_{ij}$ be the values of the variables $\W$,  $\bm{X}_{ij}$.  
At each subsequent iteration, the volume of the RCI set increases, \emph{i.e.}, $|\mathrm{det}(W^{q+1})| \geq |\mathrm{det}(W^{q})|$, if  the following LMI condition is imposed,
\begin{equation}\label{eq:Wobj}
	\W^{\top} W^{q} + (W^{q})^{\top}\W -(W^{q})^{\top}W^{q} \succcurlyeq \W_{\mathrm{obj}} \succ 0,
\end{equation}
where $\W_{\mathrm{obj}} = \W^{\top}_{\mathrm{obj}} \in \mathbb{R}^{n \times n}$ is the new symmetric matrix variable. 
Moreover, the non-linearity \eqref{eq:linearization W} can be written as,
\begin{align}\label{eq:linearization W new}
	\!\! \! \W^{\top}\bm{X}^{-1}_{ijk}\W \! \succcurlyeq \!
	\W^{\top} Z^{q}_{ijk} \!+\! {Z^{q}_{ijk}}^{\top} \W 
	\!\!-\!\! {Z^{q}_{ijk}}^{\!\!\top} \bm{X}_{ijk} Z^{q}_{ijk},
\end{align}
where $Z^{q}_{ijk} \triangleq ({X}^q_{ijk})^{-1} W^{q}$. At the $q$-th iteration, the $(1,1)$-block in \eqref{eq:Dilated LMI 1} is substituted with the right hand side of \eqref{eq:linearization W new},  
\begin{equation}\label{eq:Dilated LMI 1 new}
	\scalemath{0.9}{ \begin{bmatrix}
			\W^{\top} Z^{q}_{ijk} + {Z^{q}_{ijk}}^{\top}\W - {Z^{q}_{ijk}}^{\top} \bm{X}_{ijk} Z^{q}_{ijk} & \bm{\phi}_{ijk} C^{\top}e_{k} \\
			\bm{\phi}_{ijk} e_{k}^{\!\top}C   & \bm{\phi}_{ijk}
		\end{bmatrix} {\succcurlyeq} 0}.
\end{equation}
Thus, we solve the following SDP program at each iteration:
\textbf{Algorithm 1: $q$-th iteration:}
\begin{equation}\label{eq:iterative SDP}
	\begin{array}{lll}
		\max & \mathrm{log}\,\mathrm{det}(\W_{\mathrm{obj}}) & \\
		{\bm{Z}_{\mathrm{SDP}}} & \\
		\text{subject to:} & \eqref{eq:Wobj}, &\\
		& 
		(\ref{eq:Tractable State Constraints}),  (\ref{eq:Tractable Control Input Constraints}), &  (\text{state-input constraints}) \\
		&  \eqref{eq:Dilated LMI 2}, \ \eqref{eq:Dilated LMI 1 new}, \;  & (\text{invariance LMIs})  
	\end{array}
\end{equation}
where the optimization variables are $\bm{Z}_{\mathrm{SDP}} \triangleq \left( \W,\N,\bm{X}_{ijk},\V_{ijk},\bm{\phi}_{ijk},\bm{\Lambda}_{ijk}, \bm{\Gamma}_{ijk}, \W_{\mathrm{obj}} \right)$, for  $k \in \mathbb{I}_{1}^{n_c}$,  $i \in \mathbb{I}_{1}^{2\sigma}$, $j \in \mathbb{I}_{1}^{v_p}$.
The SDP \eqref{eq:iterative SDP} consists of 
linear state-input constraints which are identified by $2\sigma  n_x$, $2\sigma  n_u  v_p$ scalar inequalities. The invariance conditions~\eqref{eq:Dilated LMI 2},  \eqref{eq:Dilated LMI 1 new} consists of $2n_c  \sigma  v_p$ LMI constraints, each LMI has $(1+(ns+m)n+3n)$, and $(n+1)$ number of rows, respectively. 
All constraints are in terms of $3n^2 + mns + (T+1)n_w+ 1$ optimization variables. 
The number of variables scale quadratically in the state dimension $n$,  thus, the approach can be computationally expensive for systems having a large state dimension.

\section{Numerical example}\label{sec:example}
The effectiveness of the proposed method is shown via a numerical example. The  algorithm is implemented in  \texttt{Python} with \texttt{cvxpy}~\cite{diamond2016cvxpy} and  \texttt{MOSEK}~\cite{mosek}   to solve the SDP program.

\subsection*{Open-loop unstable LPV data-generating system}
We consider parameter-varying double integrator system:
\begin{align}\label{eq:example}
	x_{k+1}=\begin{bmatrix} 1+\delta_k & 1+\delta_k \\ 0 & 1+\delta_k \end{bmatrix} x_k+\begin{bmatrix} 1 \\ 1 \end{bmatrix} u_k+w_k,
\end{align}
where $|\delta_k| \leq 0.2$, with the following state-input constraints,
\begin{align*}
	\begin{bmatrix}
		0.10 \! & \! -0.10 \! &\!  0.10 \! & \! -0.10 \! &\!  0 \! & \! 0 \\
		0.15 \! & \! -0.10 \! &\!  -0.15\! &\!  0.15\!  &\!  0.25\!  &\!  -\frac{1}{6}
	\end{bmatrix}^{\! \top}\! x \leq \mathbf{1},   |u| \! \leq \!  3,
\end{align*}
The disturbance is assumed to lie in a set $\mathcal{W}=\{w:|w| \leq 0.1\}$. The system can be written in the LPV form~\eqref{eq:system} with 
\begin{align*}
	A^1=\begin{bmatrix} 1.2 & 1.2 \\ 0 & 1.2 \end{bmatrix}, A^2=\begin{bmatrix} 0.8 & 0.8 \\ 0 & 0.8 \end{bmatrix},  B=\begin{bmatrix} 1 \\ 1 \end{bmatrix}
\end{align*}
defining  $p_{k,1}=2.5(0.2+\delta_k)$, $p_{k,2}=2.5(0.2-\delta_k)$.  This corresponds to the simplex scheduling parameter set $\mathcal{P} = \{p \in \mathbb{R}^2:p \in [0,1], p_1+p_2=1\} = \mathsf{conv}(\smallmat{1\\0}, \smallmat{0\\1})$.
The system matrices $\{A^1, A^2, B\}$ are \emph{unknown} and only used to gather the data. 
A single state-input-scheduling trajectory of $T=20$ samples is gathered by exciting the system~\eqref{eq:example} with inputs uniformly distributed in $ [-3, 3]$. The data satisfies the rank condition in Proposition~\ref{prop:rich_data}, \emph{i.e}, $\mathrm{rank}(X^p_u) = 5$. 
We select the matrix $C$ to
define a regular polytope of desired complexity by assuming $\W = I$. 
The maximum volume RCI set and the associated PD
state-feedback gain matrices are computed with \textbf{Algorithm 1} solving~\eqref{eq:iterative SDP} iteratively for $5$ iterations with  data-based LMI conditions.
The obtained invariant set matrix $\W$ and feedback control gains $\{ \K^1, \K^2\}$  for the set complexity $n_c=2,3$ are :
\begin{align*}
	&C_2=\begin{bmatrix}
		1   &  0\\
		0   &  1
	\end{bmatrix},\quad \quad  \left [ \begin{array}{cc} 
		\W_2\\\hline
		\K^{1}_2 \\ \hline
		\K^{2}_2
	\end{array} \right ]=\left [ \begin{array}{cc} 
		\!\!\;\;6.02& -0.79\!\!\\
		\!\!\;\;0.02&  2.15\!\!\\\hline
		\!\!   -0.11 & -0.73\!\!\\ \hline
		-0.18 & -0.94
	\end{array} \right ], \nonumber \\
	&C_3\!=\!\!\begin{bmatrix}
		\!\;\;20\!\!  & \!\! \;\;20\!\\
		\! -20 \!\!&  \!\!  \;\;0\!\\
		\!\;\;0 \!\! &\!\! -25\!
	\end{bmatrix},\left [ \begin{array}{cc} 
		\W_3\\\hline
		\K^{1}_3 \\ \hline
		\K^{2}_3
	\end{array} \right ]\!=\!\left [ \begin{array}{cc} 
		115.82 & 44.77\\
		-14.83 &  \;\;81.46\\\hline
		-0.27 & -0.68\\ \hline
		-0.31 & -0.67\\
	\end{array} \right ]
\end{align*}

\begin{figure}[t!]
	\centering
	\begin{subfigure}[h]{0.47\columnwidth}
		\centering
		\includegraphics[width= 1\columnwidth]{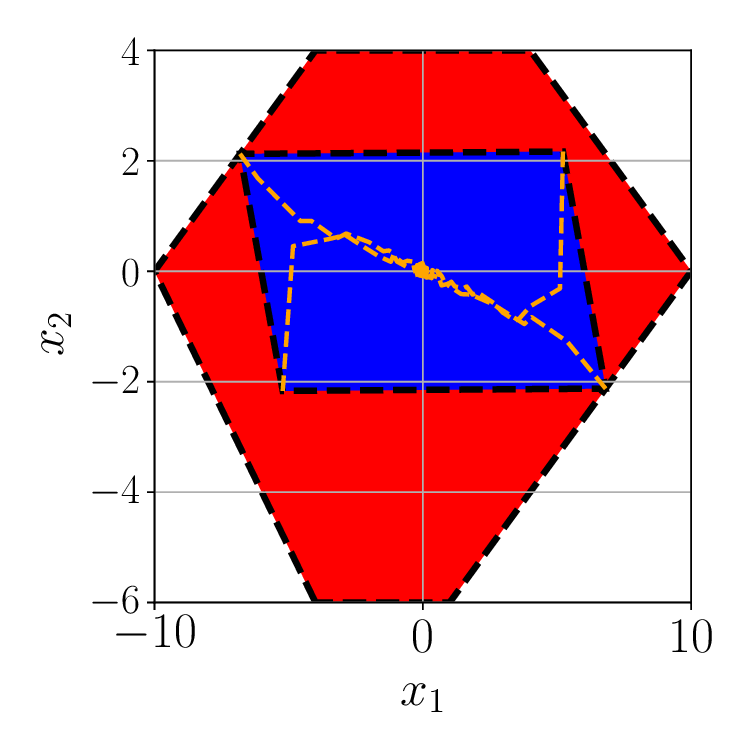}
		\vspace{-0.5cm}
		\caption{$n_c=2$ }
	\end{subfigure} 
	\begin{subfigure}[h]{0.47\columnwidth}
		\centering 
		\includegraphics[width= 1\columnwidth]{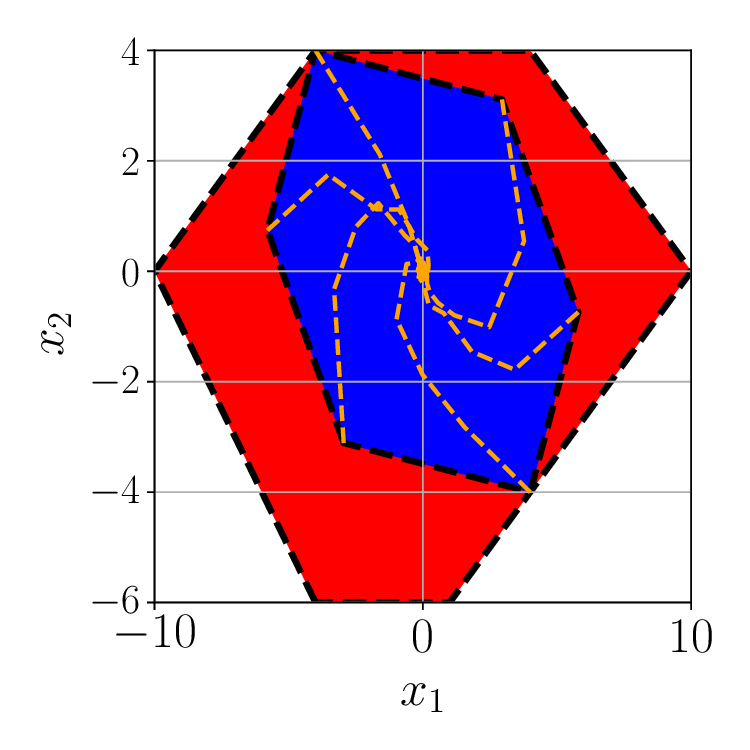}
		\vspace{-0.5cm}
		\caption{$n_c=3$ }
	\end{subfigure}
	\vspace{-0.1cm}
	\caption{Maximum volume RCI sets $\mathcal{S}$ (blue), closed-loop simulated trajectories (orange) and state constraints set (red).}
	\label{fig:rciset_P}
\end{figure}
\begin{figure}[t!]
	\centering
	\includegraphics[width=0.65\columnwidth]{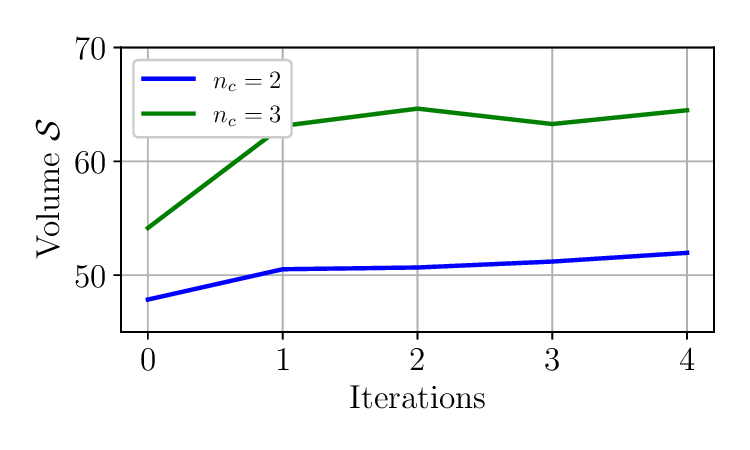}
	\vspace{-0.2cm}
	\caption{Volume of $\mathcal{S}$ \emph{vs} iterations of \textbf{Algorithm 1}}
	\label{fig:volume_iter}
\end{figure}
\begin{figure}[t!]
	\centering
	\includegraphics[width=0.65\columnwidth]{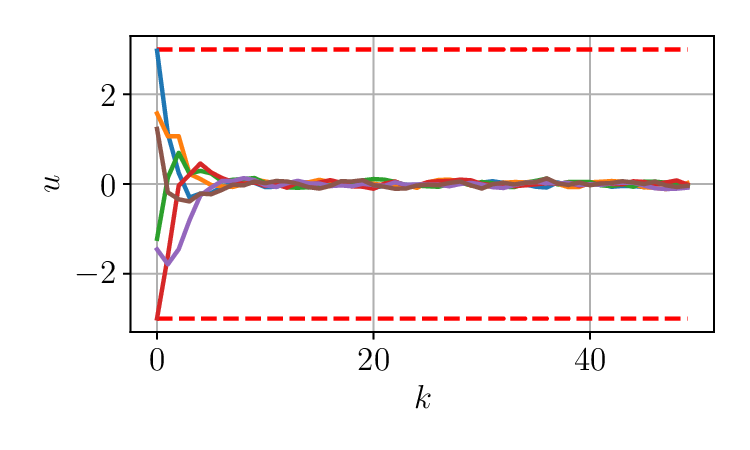}
	\vspace{-0.2cm}
	\caption{Control input $u = \K(p) x$ trajectories for the computed state-feedback gain. Input constraints (dashed-red).}
	\label{fig:input constrains}
\end{figure}
In Fig.~\ref{fig:rciset_P}, the corresponding RCI sets  are plotted. 
The  volume of the obtained RCI sets $\mathcal{S}$ are $53.14 \ (n_c=2)$ and $64.86 \ (n_c=3)$. The volume of the RCI sets w.r.t. iterations of \textbf{Algorithm 1} is shown in Fig.~\ref{fig:volume_iter}. It can be observed that as $n_c$ increases, size of the RCI set increases, thus $C$ can be used as a tuning parameter to achieve a trade-off between complexity \emph{vs} size of the set.
Fig.~\ref{fig:rciset_P} also depicts closed-loop state trajectories starting from each vertex of the RCI set and corresponding control input trajectories are shown in Fig.~\ref{fig:input constrains}. The state trajectories are obtained by simulating the true system~\eqref{eq:example} in closed-loop with the obtained state-feedback controller $u=\K(p) x$ with each vertex of the RCI set as the initial condition. Note that for each closed-loop simulation, a different realization of the scheduling signal  $p \in [0, 1]$ is generated,  
as well as   a different realization of the disturbance signal $w$ in the given bound is acting on the system at each time instance. The result shows that the approach guarantees robust invariance w.r.t. all possible scheduling signals taking values in a given set as well as in the presence of a bounded but unknown disturbance,  while respecting the state-constraints.  
Fig.~\ref{fig:input constrains} shows that the input constraints are also satisfied.
\begin{table}[thb!]
	\centering
	\begin{tabular}{ |c|c|c|c|c| } 
		\hline
		$T$ &  $20$ & $50$ & $100$ & $200$  \\ 
		\hline
		$\mathrm{volume}(\mathcal{S})$ & 64.86 & 67.04 & 68.55 & 68.77 \\
		\hline
	\end{tabular}
	\caption{Volume of $\mathcal{S}$  \emph{vs} number of data samples $T$.}\label{tab:vol_vs_T}
\end{table}
Finally, we analyse the effect of number of data samples $T$ on the size of the RCI set $\mathcal{S}$. The results are reported in Table~\ref{tab:vol_vs_T}, which show that as $T$ is increases, volume of the  RCI sets increases. This is due to the fact that the feasible model set $\mathcal{M}_{T}$ shrinks progressively $\mathcal{M}_{T+1} \subseteq \mathcal{M}_{T}$, as $T$ increases.
\section{Conclusion}
We proposed a direct data-driven method  to compute a polytopic RCI set and gain scheduled state-feedback control law for LPV system. Novel data-based sufficient invariance conditions are proposed which utilize a single state-input-scheduling trajectory without requiring to identify an LPV model of the system. The effectiveness of the proposed algorithm is shown via a numerical example to generate RCI set from a small number of collected data samples. As a future work, we propose to extend the present approach for synthesizing parameter-dependent RCI sets for LPV systems.

\section{Appendix: Proof of Lemma~\ref{lem:equivalence}} \label{sec:appendix}
\begin{proof}
	Since for each vertex $\theta^{i}, \ i \in \mathbb{I}_{1}^{2\sigma}$, and  $p^{j}, \ j \in \mathbb{I}_{1}^{v_p}$,  we have  $\theta^{i} \in \Theta, p^j \in \mathcal{P}$,  it holds that $(i) \Rightarrow (ii)$. Now we will prove  $(ii) \Rightarrow (i)$.  Any given $\theta \in \Theta$ and $p \in \mathcal{P}$ can be expressed as a convex combination of the vertices of the respective sets,
	$\theta = \sum_{i=1}^{2\sigma} \alpha_i \theta^{i}, \quad \sum_{i=1}^{2\sigma} \alpha_i = 1, \quad \alpha_i \geq 0$ and $p = \sum_{j=1}^{v_p} \beta_j p^j, \quad \sum_{j=1}^{v_p} \beta_j=1, \quad \beta_j \geq 0$
	Then, based on the closed-loop dynamics \eqref{eq:closeloop} we get,
	\begin{align}
		&\theta^{+} \!=\! \W^{-1}  M\begin{bmatrix}
			\bar{\W} \\ 
			\N 
		\end{bmatrix} \left[\sum_{j=1}^{v_p} \beta_j p^j \otimes  \sum_{i=1}^{2\sigma}\alpha_i \theta^i \right]  + \W^{-1}w \label{eq:cmb0} \\
		& = \! \sum_{j=1}^{v_p} \! \beta_j \! \sum_{i=1}^{2\sigma} \! \alpha_i \underbrace{ \left(\W^{-1} M \begin{bmatrix}
				\bar{\W} \\ 
				\N 
			\end{bmatrix} \left[p^j \otimes   \theta^i \right]  + \W^{-1}w \right)}_{{\theta^{i,j}}^{+} \in \Theta} \label{eq:cmb1} \\
		&=\sum_{j=1}^{v_p} \beta_j  \underbrace{ \sum_{i=1}^{2\sigma} \alpha_i {\theta^{i,j}}^{+}}_{{\theta^{j}}^{+} \in \Theta} =\sum_{j=1}^{v_p} \beta_j  {\theta^{j}}^{+} \in \Theta, \label{eq:cmb3}
	\end{align}
	where \eqref{eq:cmb1} follows from the distributive property of the Kronecker product. 
	We know that ${\theta^{i,j}}^{+} \in \Theta$ according to \eqref{eq:inv_cond_vertices}. As ${\theta^{j}}^{+}$ in \eqref{eq:cmb3} is  a convex combination of ${\theta^{i,j}}^{+}$ and  the set $\Theta$ is convex, we obtain ${\theta^j}^{+} \in \Theta $. Similarly, as $\theta^{+}$  is  a convex combination of ${\theta^{j}}^{+} \in \Theta$ and  the set $\Theta$ is convex, we get $\theta^{+} \in \Theta$, thus proving $(ii) \Rightarrow (i)$.
\end{proof}

\bibliographystyle{plain}
\bibliography{lcss23_ref}
\end{document}